\documentclass[twocolumn]{article}
\usepackage[margin=2.5cm]{geometry}
\usepackage{amsmath,amssymb,amsthm,xcolor}
\usepackage[
backend=biber,
style=numeric,
sorting=nyt,
sortcites,
giveninits=true,
url=false,
doi=false,
isbn=false,
date=year,
maxbibnames=4,
]{biblatex}
\renewbibmacro{in:}{}
\DeclareFieldFormat{pages}{#1}
\renewbibmacro*{volume+number+eid}{%
  \printfield{volume}%
  \setunit*{\addnbthinspace}
  \printfield{number}%
  \setunit{\addcomma\space}%
  \printfield{eid}}
\DeclareFieldFormat[article]{number}{\mkbibparens{#1}}
\usepackage[hidelinks]{hyperref}

\newtheorem{thm}{Theorem}
\newtheorem*{Penthm}{Penrose Singularity Theorem}

\newtheorem{lem}[thm]{Lemma}

\theoremstyle{definition}

\newcommand{\Ric}{\operatorname{Ric}}
\newcommand{\Riem}{\operatorname{Riem}}

\newcommand{\trap}{\mathcal{T}}

\title{Radial Gravitational Collapse Causes Timelike Incompleteness}
\author{Leonardo García-Heveling \footnote{Radboud University Nijmegen, the Netherlands (former) and Universit\"at Hamburg, Germany (present). Email: \texttt{leonardo.garcia@uni-hamburg.de}}}

\begin{document}

\maketitle

\begin{abstract}
We show that a globally hyperbolic spacetime containing a trapped surface and satisfying the strong energy condition and a condition on certain radial tidal forces must be timelike geodesically incomplete. This constitutes a ``timelike" version of Penrose's celebrated singularity theorem. Recall that the latter concludes that certain spacetimes are null incomplete, providing the first theoretical evidence that black holes actually exist in our Universe. By concluding timelike instead of null incompleteness, we obtain, at the expense of stronger assumptions, a clearer physical interpretation and the existence of an event horizon.
\end{abstract}

\section{Introduction}

In his seminal 1964 article \cite{PenSing}, Roger Penrose showed that gravitational collapse in General Relativity leads to the formation of spacetime singularities. At the time, some singular solutions to Einstein's Equations were already known, such as the Schwarzschild solution. It was believed, however, that  such singularities were mathematical pathologies arising due to the high degree of symmetry and the lack of matter, and that physically realistic solutions would be singularity free. Penrose forever changed this point of view through his famous singularity theorem.

\begin{Penthm}
    Let $(M,g)$ be a spacetime containing a non-compact Cauchy surface. Assume that the null convergence condition is satisfied (meaning $\Ric(X,X) \geq 0$ for every null vector $X$) and that $(M,g)$ contains a trapped surface $\trap$. Then $(M,g)$ is null geodesically incomplete.
\end{Penthm}

The key concept of \emph{trapped surface} was introduced for the first time through this theorem. The original definition is that $\trap$ is a codimension $2$ spacelike submanifold with negative ingoing and outgoing null expansion. Physically, it corresponds to the boundary of a spatial region from which even light cannot escape (hence the name). The incomplete null geodesic predicted by the theorem corresponds to the trajectory of a light ray which ends abruptly, the interpretation being that it must have met a spacetime singularity. Since the remaining assumptions are rather mild, it provided the first compelling evidence that the formation of singularities is in fact a crucial feature of General Relativity \cite{LanCC}.

The nature of the singularity predicted by Penrose's theorem, however, remains mysterious to this day. While it is widely believed to be at the center of a black hole, the theorem itself gives no information about features usually associated to such objects: Is it shielded by an event horizon? Does the curvature blow up? Penrose himself was aware of the problem and conjectured some of these properties, but these so-called cosmic censorship conjectures remain open to this day \cite{LanCC}.

Another problem is that Penrose's theorem only predicts null incompleteness. It thus remained open whether massive objects or even human observers, which follow timelike curves, can fall into the singularity. Moreover, the affine parameter of a null geodesic (i.e.\ the quantity that characterizes incomplete null geodesics) does not correspond to a physical observable. The affine parameter of a \emph{timelike} geodesic, on the other hand, is the proper time measured by an observer.

The issue of timelike vs null incompleteness had already been identified early on by Geroch \cite{GerSing} and Beem \cite{Beem}, who gave some negative results, showing that null incompleteness does not always imply timelike incompleteness. The counterexamples given, however, are physically rather artificial. In this paper, we reopen the case by proving a positive result.

\begin{thm} \label{thm:main}
    Let $(M,g)$ be an $n+1$-dimensional spacetime. Assume that
    \begin{enumerate}
     \item $(M,g)$ is globally hyperbolic.
     \item $\Ric(X,X) \geq 0$ for every timelike vector $X$.
     \item $(M,g)$ contains a trapped surface $\trap$ with (necessarily past-pointing timelike) mean curvature vector denoted by $H$ and $H_0 := \min \sqrt{\vert g(H,H) \vert}$.
     \item There exists a constant $0<k< H_0$ such that every unit speed timelike geodesic $\sigma$ normal to $\trap$ of length greater than $\frac{1}{k}$ satisfies
     \begin{equation} \label{eq:initialcond}
      \int_0^{\frac{1}{k}} \left(1-k^2u^2 \right) \Riem(\dot\sigma,U,\dot\sigma,U) du \leq 0,
     \end{equation}
     where $U$ is the parallel transport along $\sigma$ of the unique (up to sign) spacelike unit vector orthogonal to $\dot\sigma(0)$ and $\trap$.
    \end{enumerate}
    Then every causal curve starting on $\trap$ has length at most $\ell := \frac{1}{k}+\frac{n}{H_0-k}$. In particular, $(M,g)$ is timelike geodesically incomplete.
\end{thm}

Note that our theorem predicts that \emph{every} observer starting on $\trap$ will fall into the singularity after a proper time of at most $\ell$. From this, we can deduce the existence of an event horizon, because it means that an observer reaching a trapped surface can no longer escape to infinity. We return to this important point in the conclusions section. For now, let us comment on each of the assumptions:
\begin{enumerate}
 \item Global hyperbolicity is equivalent to the existence of a Cauchy surface \cite{GerDD,BeSa1}, which can be compact or non-compact (unlike in Penrose's theorem).
 \item This is the \emph{timelike convergence condition}, which is stronger than the null convergence condition. Through the Einstein equations it can be translated into a bound on the energy-momentum tensor known as the \emph{strong energy condition}, which is satisfied by many forms of matter \cite[pp.~95-96]{HaEl}.
 \item A trapped surface $\trap$ is a compact, orientable, spacelike codimension $2$ submanifold with everywhere past-pointing timelike mean curvature vector $H$. This definition is equivalent to the usual one in terms of negative inward and outward null expansions \cite[p.~435]{ONeill}. Note that by compactness $H_0>0$.
 \item The sectional curvatures $\Riem(\dot\sigma,U,\dot\sigma,U)$ have the physical interpretation of tidal forces. In our case, $U$ should be interpreted as the radial direction pointing inward (or outward) of the trapped surface. We are thus requiring that the radial tidal forces are repulsive, at least for some proper time $\frac{1}{k}$ in a weighted average sense (where the weight decreases toward the future).
\end{enumerate}

Assumption (iv) is justified by the physical intuition that the gravitational attraction grows as one approaches the source, causing a test object freely falling towards a star to be stretched \cite[§32.6]{MTW}. In Schwarzschild spacetime, (iv) is satisfied, and in fact the radial tidal forces even blow up to $-\infty$. In Kerr spacetime, however, (iv) can be violated for geodesics staying close to the axis of rotation of the black hole \cite{ONeillKerr}. Thus further investigation to allow the case of rotating black holes is needed (see also the discussion section at the end). From a conceptual point of view, assumption (iv) also has the disadvantage that it is an assumption on the geometry of our particular spacetime. In comparison, the energy condition (ii) depends only on the matter model considered, but not on the specific solution. By framing (iv) as an integral assumption (rather than pointwise), we do, luckily, avoid the need for extreme fine tuning.

We prove the theorem in the next section, and conclude the paper with a discussion about event horizons and possibilities to weaken our assumptions.

\section{Proof of main theorem}

We use techniques from the book of O'Neill \cite{ONeill}, which are used there to prove the Hawking singularity theorem. The latter is similar to the Penrose singularity theorem, but its assumptions and conclusions are usually interpreted from a cosmological point of view, instead of in the context of gravitational collapse.

We start with the global causality theoretic argument of our proof. To that end, recall that the Lorentzian length $L_g(\gamma)$ of a causal curve $\gamma \colon [0,b] \to M$ is given by
\begin{equation}
 L_g(\gamma) := \int_0^b g(\dot\gamma,\dot\gamma) du,
\end{equation}
and equals $b$ when $\gamma$ is parametrized with unit speed.

\begin{lem} \label{lem:global}
 For every point $q$ in the timelike future of $\trap$, there exists a maximizing timelike geodesic $\sigma$ from $\trap$ to $q$. Moreover, $\sigma$ is normal to $\trap$.
\end{lem}

\begin{proof}
 Because $(M,g)$ is globally hyperbolic, the Lorentzian distance from $q$ (to the past),
 \begin{equation}
  \tau(x,q) := \sup \left\{ L_g(\sigma) \mid \sigma \text{ causal from } x \text{ to } q \right\},
 \end{equation}
is continuous \cite[Lem.~14.21]{ONeill}. Thus $\tau(\cdot,q)$ attains a maximum on the compact set $\trap$. For every point $p \in \trap$ where said maximum is attained, again by global hyperbolicity there exists a geodesic $\sigma$ from $p$ to $q$ with $L_g(\sigma) = \tau(p,q)$ \cite[Lem.~14.19]{ONeill}. That $\sigma$ must be normal to $\trap$ is standard \cite[Cor.~10.26]{ONeill}.
\end{proof}

Our goal is to show that the length of all maximizing geodesics provided by Lemma \ref{lem:global} is smaller than $\ell$, from which the statement of Theorem \ref{thm:main} follows immediately, since all other curves must be shorter than the maximizing one (here $\ell := \frac{1}{k}+\frac{n}{H_0-k}$ as in the theorem). It suffices to show that if there was a normal geodesic $\sigma \colon [0,\ell] \to M$ of length $\ell$, it would have a focal point $\sigma(r)$ for $r< \ell$ (defined below), and hence would not be maximizing.

Let $\sigma \colon [0,\ell] \to M$ be a unit-speed geodesic normal to $\trap$. We need to consider infinitessimal variations of $\sigma$ that remain normal to $\trap$. To each such variation corresponds a \emph{variation vector field} $V \colon [0,\ell] \to TM$ along $\sigma$ such that $V(0)$ is tangent to $\trap$ and $V(\ell)=0$, and visceversa any such variation vector field gives rise to an infinitessimal variation of $\sigma$. The intuition is that we are pushing $\sigma$ along $V$, while keeping the starting direction normal to $\trap$ and the endpoint fixed. Because $\sigma$ is a normal geodesic, the first derivative of the length for these variations vanishes. The second variation $L''_V$ of the length is then given by Synge's formula
\begin{align} \label{eq:secondvar}
   L''_V = &-\int_0^\ell \left( g(\dot V, \dot V) - \Riem(\dot\sigma,V,\dot\sigma,V) \right) du \nonumber \\ &+ g(\dot\sigma(0),II(V(0),V(0))),
\end{align}
where dots mean differentiation along $\sigma$ and $II$ is the second fundamental form of $\trap$ (see \cite[Chap.~10]{ONeill}). A focal point of $\trap$ along $\sigma$ is a point $\sigma(r)$ such that there exists a variation vector field $V$ with $V(r) = 0$ and $L''_V \geq 0$. When there is a focal point $\sigma(r)$, then $\sigma$ no longer maximizes the distance to $\trap$ after $\sigma(r)$.

\begin{lem}
 Every unit-speed timelike geodesic $\sigma \colon [0,\ell] \to M$ normal to $\trap$ encounters a focal point $\sigma(r)$ with $0 < r \leq \ell$, provided $\sigma$ is defined on $[0,\ell]$.
\end{lem}

\begin{proof}
 Fix any point $p \in \trap$. Let $e_1=U$ and let $e_{2},...,e_{n}$ be an orthonormal basis for $T_p \trap$. Parallel transport $e_1,...,e_{n}$ along $\sigma$, obtaining vector fields $E_1,...,E_{n}$. In particular, $E_1 = U$ (cf.\ assumption (iii)). For $u \in [0,\ell]$ the affine parameter of $\sigma$, and $k$ as in assumption (iv), define
 \begin{align}
  \varphi(u) &:= \begin{cases} k u &\text{for } u < 1/k, \\ 1-\left(u-\frac{1}{k}\right)\frac{H_0-k}{n} &\text{for } u \geq 1/k,\end{cases}
  \\
  \psi(u) &:= \begin{cases} 1 &\text{for } u < 1/k, \\ 1-\left(u-\frac{1}{k}\right)\frac{H_0-k}{n} &\text{for } u \geq 1/k.\end{cases}
 \end{align}
Note that $\psi(0) =1$ while $\varphi(0) = \varphi(\ell)=\psi(\ell)=0$. We use these as coefficients to define $n$ variation vector fields in the following way:
\begin{align}
 &V_1 := \varphi E_1, &V_j := \psi E_j \quad &\text{for} \quad j = 2,...,n.
\end{align}
Note here that $\varphi, \psi$ are not smooth, but we can smoothly approximate them (more to that below). By \eqref{eq:secondvar} we have
\begin{align}
 L''_{V_1} = &- \int_0^{\ell} \left((\varphi')^2 - \varphi^2 \Riem(E_1, \dot\sigma, E_1, \dot\sigma) \right) du \\
 = &-k - \frac{H_0-k}{n} + \int_0^\ell \psi^2 \Riem(E_1, \dot\sigma, E_1, \dot\sigma) du \nonumber \\ &- \int_0^\ell \left(\psi^2 - \varphi^2\right) \Riem(E_1, \dot\sigma, E_1, \dot\sigma) du \\
 = &-k - \frac{H_0-k}{n} + \int_0^\ell \psi^2 \Riem(E_1, \dot\sigma, E_1, \dot\sigma) du \nonumber \\ &- \int_0^{\frac{1}{k}} \left( 1-k^2 u ^2\right) \Riem(E_1, \dot\sigma, E_1, \dot\sigma) du
\end{align}
and for $j = 2,...,n$,
\begin{align}
 L''_{V_j} = &- \int_0^{\ell} \left((\psi')^2 - \psi^2 \Riem(E_j, \dot\sigma, E_j, \dot\sigma) \right) du \nonumber \\ &+ g \left( \dot\sigma(0), II(e_j,e_j) \right) \\
 = &- \frac{H_0-k}{n} + \int_0^{\ell} \psi^2 \Riem(E_j, \dot\sigma, E_j, \dot\sigma) du  \nonumber \\ &+ g \left( \dot\sigma(0), II(e_j,e_j) \right)
\end{align}
Summing these two expressions, and using that
\begin{align}
 \Ric(\dot\sigma,\dot\sigma) &= \sum_{i=1}^n \Riem(E_i,\dot\sigma,E_i,\dot\sigma), \label{eq:trace} \\ &H = \sum_{j = 2}^n II(e_j,e_j),
\end{align}
and assumption (ii)--(iv), we obtain
\begin{align}
 \sum_{i = 1}^{n} L''_{V_i} = &-(H_0-k) + \int_0^\ell \psi^2 \Ric(\dot\sigma, \dot\sigma) du  \nonumber \\ &+ g(\dot\sigma(0),H) - k  \nonumber \\ &- \int_0^{\frac{1}{k}} \left( 1-k^2 u ^2\right) \Riem(E_1, \dot\sigma, E_1, \dot\sigma) du \\
 > &-H_0 + g(\dot\sigma(0),H) \geq 0,
\end{align}
where the last inequality follows from the reverse Cauchy--Schwartz inequality \cite[Prop.~5.30]{ONeill} together with the fact that $H, \sigma(0)$ are past- and future-directed timelike respectively. We conclude that at least one of the $L''_{V_i} $ must be positive, and hence $\sigma$ has a focal point, as desired. Since the inequality is strict, the argument also works for smooth approximations of $\varphi,\psi$ close enough in the $C^1$-topology on the space of continuous functions $[0,\ell]$ (this is just a technicality needed because the theory developed in \cite{ONeill} assumes smoothness).
\end{proof}

\section{Discussion}

At the expense of strengthening the assumptions, we have been able to improve on the conclusions of the Penrose singularity theorem in two important aspects:
\begin{itemize}
    \item We conclude timelike incompleteness, which (unlike null incompleteness) has a clear physical interpretation in terms of the proper time experienced by observers.
    \item We can further conclude the existence of an event horizon, since once an observer reaches a trapped surface, they cannot escape anymore and are doomed to fall into the singularity after a proper time of at most $\ell$. This is true even for accelerated observers.
\end{itemize}
Here by horizon we simply mean the boundary that separates the regions $\mathcal{B}$ and $\mathcal{E}$ defined as follows. We call \emph{black hole interior} $\mathcal{B}$ the region such that every future-directed timelike curve starting in $\mathcal{B}$ is incomplete, while by \emph{black hole exterior} $\mathcal{E} := M \setminus \mathcal{B}$ we mean its complement, implying that from every point in $\mathcal{E}$, one can find at least one future-directed timelike curve of infinite length (i.e.\ an observer that lives forever). Under the assumptions of our theorem, $\mathcal{B} \neq \emptyset$, and $\mathcal{E}$ should also be assumed non-empty, unless we are in a Big Crunch scenario where the whole Universe is swallowed by the singularity. Incidentally, Senovilla \cite{SenUM} has argued that in certain situations, gravitational collapse \emph{can} lead to a Big Crunch, but we shall ignore this in the present discussion. The general treatment of event horizons as described here was pioneered by M\"uller \cite{MueHor}, and has the advantages over the traditional approach using null infinity that we do not need to assume asymptotic flatness, or to worry about making a choice of spacetime compactification.

A word of caution about the the horizon predicted by our theorem: It is in principle possible that there exist some timelike geodesics (or, more generally, curves with bounded acceleration) which are incomplete despite not emanating from a trapped surface. Our theorem gives us no control over those. In particular, they might never cross the horizon as we have defined it (meaning they could stay in $\mathcal{E}$ and never enter $\mathcal{B}$). Thus while the singularities \emph{which our theorem predicts} are hidden behind a horizon, there could be other singularities in our spacetime (or the predicted singularity could be ``larger''). Recall that Penrose's weak cosmic censorship conjecture states states that \emph{every} singularity in a physically reasonable spacetime should be hidden behind an event horizon.

To conclude the paper, let us mention that many variations of the singularity theorems of Hawking and Penrose are known. These have relaxed assumptions, particularly the causality assumptions \cite{LesSing,MinSing}, the energy conditions \cite{GKOS,PaeSing,FeKoSing}, the regularity of the spacetime metric \cite{GraSing,KOSS}, and the dimensionality of the trapped submanifold \cite{GaSeSing} (this list of references is in no way exhaustive). The goal is often to accommodate for violations in the classical assumptions due to quantum phenomena. Not only should many of these modifications be possible for our theorem, but in fact it paves the way for further progress. On the one hand, our theorem is in the spirit of Penrose in that it is about gravitational collapse. On the other hand, our proof is more similar to that of the Hawking singularity theorem, which has generally proven to be easier to modify. Thus the chance arises to adapt improved variants of Hawking's theorem to the context of gravitational collapse.

We have, however, also created the need to relax our new assumption (iv) (which has no analogue in the original singularity theorems of Hawking and Penrose). From our proof it is already possible to derive some sufficient conditions that are technically weaker than (iv), but they make the statement of the theorem more convoluted without improving on its physical interpretation (for example, we can allow $k$ to depend on the geodesic, as long as it remains uniformly bounded away from $H_0$). More interesting (at least from a mathematical point of view) is to substitute (ii) and (iv) together by an intermediate curvature bound. By this we mean a bound on the trace of the Riemann tensor over a $n-1$ dimensional subspace (instead of the full $n$ dimensions that yield the Ricci tensor, cf.\ \eqref{eq:trace}). Then we can replicate our proof but completely ignoring the ``radial'' direction $E_1$. But while the Ricci bound $\Ric(X,X) \geq 0$ is equivalent to the strong energy condition via the Einstein Equations, this is not true for intermediate curvature bounds, leaving their physical meaning in the air. Note that such an approach was taken by Galloway and Senovilla \cite{GaSeSing}, but there the directions left out are to be interpreted as compactified string theory dimensions, which should not contribute to macroscopic physics anyhow, while in our case leaving out the radial direction seems artificial. Satisfactorily weakening assumption (iv) thus remains as an open point for further investigation.

\bigskip
The author thanks Melanie Graf for useful discussions.

\printbibliography

\end{document}